\documentclass[a4paper,UKenglish,cleveref, autoref, thm-restate]{lipics-v2021}
\hideLIPIcs

\usepackage{amsmath,amssymb,amsthm}
\usepackage{lineno}
\usepackage{tikz,cite}
\usepackage{thm-restate}
\usepackage{graphicx}
\usepackage{float}
\usepackage{xcolor}[dvipsnames]

\usepackage{cleveref}

\title{Twin-width one}

\author{Jungho Ahn}{Korea Institute for Advanced Study (KIAS), Seoul, South~Korea \and \url{https://www.junghoahn.com/} }{junghoahn@kias.re.kr}{https://orcid.org/0000-0003-0511-1976}{Supported by the KIAS Individual Grant (CG095301) at Korea Institute for Advanced Study.}

\author{Hugo Jacob}{LIRMM, CNRS, Université de Montpellier, France}{hugo.jacob@lirmm.fr}{https://orcid.org/0000-0003-1350-3240}{}

\author{Noleen K\"ohler}{
University of Leeds, UK}{n.koehler@leeds.ac.uk}{https://orcid.org/0000-0002-1023-6530}{}

\author{Christophe Paul}{LIRMM, CNRS, Université de Montpellier, France}{paul@lirmm.fr}{https://orcid.org/0000-0001-6519-975X}{}

\author{Amadeus Reinald}{LIRMM, CNRS, Université de Montpellier, France}{amadeus.reinald@lirmm.fr}{https://orcid.org/0000-0002-8108-4036}{}

\author{Sebastian Wiederrecht}{School of Computing, KAIST, Daejeon, South Korea}{wiederrecht@kaist.ac.kr}{https://orcid.org/0000-0003-0462-7815}{}

\authorrunning{J. Ahn, H. Jacob, N. K\"{o}hler, C. Paul, A. Reinald, and S. Wiederrecht} 

\Copyright{Jungho Ahn, Hugo Jacob, Noleen K\"{o}hler, Christophe Paul, Amadeus Reinald, and Sebastian Wiederrecht} 



\acknowledgements{The authors wish to thank the organisers of the 1st Workshop on Twin-width, which was partially financed by the grant ANR ESIGMA (ANR-17-CE23-0010) of the French National Research Agency. The second author wishes to thank Paul Bastide and Carla Groenland for interesting initial discussions on the topic.}

\nolinenumbers 

\newenvironment{subproof}[1][\proofname]{
    
    \begin{proof}[#1]}{\end{proof}
}

\newcommand{\tww}{\operatorname{tww}}

\begin{document}
\maketitle

\begin{abstract}
    We investigate the structure of graphs of twin-width at most $1$, and obtain the following results:
    \begin{itemize}
        \item Graphs of twin-width at most $1$ are permutation graphs. In particular they have an intersection model and a linear structure.
        \item There is always a $1$-contraction sequence closely following a given permutation diagram.
        \item Based on a recursive decomposition theorem, we obtain a simple algorithm running in linear time that produces a $1$-contraction sequence of a graph, or guarantees that it has twin-width more than $1$.
        \item We characterise distance-hereditary graphs based on their twin-width and deduce a linear time algorithm to compute optimal sequences on this class of graphs.
    \end{itemize}
\end{abstract}

\section{Introduction}

Twin-width is a graph invariant introduced by Bonnet, Kim, Thomassé, and Watrigant \cite{twin-width1} as a generalisation of a parameter on permutations introduced by Guillemot and Marx \cite{GuillemotM14}.
The main result of the seminal paper \cite{twin-width1} is an FPT algorithm for FO model-checking on graphs of bounded twin-width, when given a \emph{contraction sequence} certifying their width. Graphs of bounded twin-width capture a wide variety of classes such as bounded rank-width and proper minor-closed classes, unifying many results on tractable FO model-checking. Furthermore, it gives an exact dichotomy between tractability and intractability of FO model-checking for hereditary classes of ordered structures \cite{OrderedGraphs} and of tournaments \cite{Tournaments}.

While these results are sufficient to consider twin-width  an important graph parameter, we are still lacking an FPT algorithm for computing `good' contraction sequences, i.e. contraction sequences of width bounded by a function of the twin-width. Such an algorithm is required for FO model-checking to be FPT on bounded twin-width classes. 
Regarding exact algorithms for twin-width, it is known that computing twin-width is hard even for constant values: distinguishing graphs of twin-width $4$ from graphs of twin-width $5$ is NP-hard \cite{twwNPhard}.
Still, a polynomial-time algorithm for recognizing graphs of twin-width $1$ was given in \cite{BonnetKRTW22}, where the worst-case time complexity is not explicitly given but corresponds to a polynomial of degree~$4$ or~$5$ depending on implementation. Meanwhile, the hardness of recognizing twin-width~$2$ and~$3$ graphs remains open. Nevertheless, it is known that sparse graphs of twin-width~$2$ have bounded treewidth \cite{BergougnouxGGHP23}.
As for approximation, there is an algorithm to compute contraction sequences of approximate width for ordered structures \cite{OrderedGraphs}. It is still wide open whether there is an XP, let alone FPT, approximation algorithm in the general case. Regarding parameterized algorithms, some results are known with parameters that are still quite far from twin-width \cite{BalabanGR24,TwwVertexIntegrity}. Interestingly, the more general parameter flip-width introduced by Torunczyk \cite{Torunczyk23} has an XP approximation algorithm but the status of FO model-checking for this parameter is unknown.

In this paper, we are interested in the structure of graphs of twin-width $1$, and in the complexity of their recognition.
Let us briefly recall the definition of $k$-contraction sequences, which are witnesses for twin-width at most $k$. Given a graph $G$, a contraction sequence of $G$ starts with $G$, and consists in a sequence of identifications of pairs of vertices (contractions) recording `neighbourhood errors' as \emph{red edges}, ending with a graph on a single vertex. Specifically, at each step, we create red edges from the contracted vertex to vertices that were not homogeneous to the pair. That is, vertices which were not both adjacent or non-adjacent to the pair. In particular, red edges remain red. Then, a \emph{$k$-contraction sequence} is one where at each step the vertices have most $k$ incident red edges.

Towards understanding graphs of twin-width one, it will be useful to look at the structure and behaviour of classes of graphs which are related. There is a rich literature on hereditary classes of graphs related to perfect graphs. It follows from the following observation and the strong perfect graph theorem \cite{strongperfectthm} that twin-width $1$ graphs are perfect graphs. It is then natural to wonder how they compare to other subclasses of perfect graphs.
\begin{observation}[See~{\cite[Lemma~2.3]{gridtww}}]
    Every cycle of length at least~$5$ has twin-width~$2$. This also implies that complements of cycles of length at least~$5$ have twin-width~$2$.
\end{observation}
A natural way of understanding classes of twin-width at most $k$ is to interpret them as a generalization of cographs. Indeed, cographs are exactly the graphs which admit a sequence of twin identifications ending in a single vertex, that is, graphs of twin-width $0$. Then, graphs of twin-width $k$ correspond to relaxing the twin condition to allow for at most $k$ `twin errors'.

We now compare twin-width $1$ graphs with other generalisations of cographs.
Cographs are the graphs of clique-width at most $2$. In this direction, cographs are generalised by the class of distance-hereditary graphs which are themselves closely related to graphs of clique-width at most $3$ \cite{CorneilHLRR12}. Distance-hereditary (DH) graphs are not closed by complementation, and have a tree-like structure. They also have a characterisation by a sequential elimination of twins and pendant vertices. 
Cographs are permutation graphs, which is exactly the complementation-closed class of graphs whose edges can be transitively oriented \cite{DushnikMiller}. Permutation graphs are asteroidal triple-free (AT-free), as such, they are dominated by a path \cite{Corneil1997}. A caterpillar is an AT-free tree , the following result suggests a relation between AT-free graphs and twin-width~$1$.

\begin{lemma}[Ahn, Hendrey, Kim, and Oum~{\cite[Lemma~6.6]{AhnHKO22}}]\label{lem:tww-tree}
    For a tree~$T$, $\tww(T)\leq1$ if and only if~$T$ is a caterpillar.
\end{lemma}

Cographs, distance-hereditary graphs, and permutation graphs can all be recognised in linear time. It is only natural to expect that the closely related class of twin-width $1$ graphs can also be recognised efficiently.

\paragraph*{Previous results on graphs of twin-width at most $1$}

The recognition algorithm for twin-width one given in \cite{BonnetKRTW22} makes use of the relation between modules and twin-width, by observing that the twin-width of a graph can be determined by considering independently the subgraphs induced by modules (see Lemma~\ref{lem:seq-modular}).
It is straightforward to deduce that the twin-width of a graph is the maximal twin-width over all prime nodes of its modular decomposition (defined in \cref{sec:prelim}). Since this decomposition may be computed in linear time \cite{McConnellS99}, one then only needs to recognise prime graphs of twin-width one. In their paper \cite{BonnetKRTW22}, the recognition is done by branching on valid sequences that have at most one red edge at any intermediate step of the sequence, after the observation that such sequences always exist. Furthermore, when the trigraph has a red edge, the only possible contractions in such a sequence involve at least one vertex of the red edge.
Essentially, the complexity of their algorithm stems from the need to branch over all possible first contractions, to greedily simulate a contraction sequence for each possible first contraction, and the lack of an efficient way to detect eventual twins to be contracted at each step.
In particular, a more explicit understanding of the structure of graphs of twin-width at most $1$ seems to be required to avoid the enumeration of all pairs of possible first contractions.

\paragraph*{Our results}
We investigate the structure of graphs of twin-width at most $1$, and uncover that it defines a relatively well-behaved hereditary class of graphs, which fits surprisingly well in the landscape of classical hereditary classes.
Graphs of twin-width at most~$1$ were already known to have some form of linear structure, due to a bound on their linear rankwidth in \cite{BonnetKRTW22}. We show that they are in fact contained in the class of permutation graphs, thus they admit an intersection model called `permutation diagram', making it easier to reason on their structure. Furthermore, we show that there is a close relationship between permutation diagrams and $1$-contraction sequences. We use these structural results to obtain a recursive decomposition theorem for prime graphs of twin-width at most~$1$. We then use it to devise a simple linear time algorithm to compute a $1$-contraction sequence, or conclude that the graph has twin-width more than~$1$. The algorithm starts by computing a permutation diagram and a modular decomposition which can be done in time $O(n+m)$, where~$n$ is the number of vertices and~$m$ is the number of edges, and then it checks that the diagram respects the decomposition theorem in time~$O(n)$.

Regarding the challenge of avoiding the enumeration of possible first contractions, we can significantly reduce the number of candidates by using the permutation diagram and related properties of contraction sequences. We consider a slightly different scheme to avoid this challenge. Instead of guessing the first contraction, we guess the vertex that will be contracted last. It turns out that, in well-behaved contraction sequences, such a vertex holds a special position in permutation diagrams, reducing the number of candidates to $4$. 
However, we do not exclude the possibility that the scheme of the previous algorithm can be implemented in linear time with some further arguments.

Interestingly, our structural analysis via the permutation diagrams allows us to avoid a characterisation by forbidden induced subgraphs. Moreover, we prove most structural tools inductively and in a unified way instead of proving the observations required for our algorithm separately.

We also present some intermediate results on distance-hereditary graphs and split decompositions (defined in \cref{sec:DH}) that turned out to be unnecessary for the recognition algorithm thanks to the use of permutation diagrams. We show that twin-width $1$ trigraphs with a pendant red edge are distance-hereditary. We also show that distance-hereditary graphs have twin-width at most $2$. Combining these results, we can compute an optimal contraction sequence for distance-hereditary graphs in linear time. We also obtain a simple inductive proof that all distance-hereditary AT-free graphs are permutation graphs, and a characterisation of the twin-width of a distance-hereditary graph by the structure of its split decomposition. This is a generalisation of the characterisation of twin-width~$1$ trees \cite{AhnHKO22} mentioned above.

Distance-hereditary graphs are also related to the following infinite family of permutation graphs of twin-width $2$: linking two obstructions to distance-hereditary graphs (domino, house, or gem) by a path of arbitrary length. In particular, this class of examples shows that a linear time algorithm for the recognition of twin-width $1$ graphs cannot be deduced from the linear algorithm to find patterns in a permutation of Guillemot and Marx \cite{GuillemotM14}.

We also observe that bipartite graphs of twin-width $1$ constitute a well-behaved subclass. Indeed, in their $1$-contraction sequences, all trigraphs are bipartite. This is because there are no induced odd cycles of length more than $3$, and triangles have at most one red edge from which we can deduce an original edge and complete it in a triangle of the starting graph (this is a particular case of \cref{obs:trigraphs-induced}). They constitute a subclass of bipartite permutation graphs whose twin-width is at most 2. Indeed, bipartite permutation graphs and proper interval graphs admit universal graphs that are unions of independent sets (resp. cliques) with half graphs between consecutive independent sets (resp. cliques), see \cite{BrandstadtL03,Lozin2011}. Both universal graphs are easily seen to admit a contraction sequence where the red graph remains a path. As a consequence we can also compute twin-width on these two classes.

\paragraph*{Perspectives}

One may hope that understanding the structure of graphs of twin-width $1$ well enough can guide a subsequent study of graphs of twin-width $2$. We expect that this hereditary class can be characterised by a tree-like decomposition with a very constrained local structure. Our results also point toward the fact that split decompositions could be a convenient tool for the analysis of the structure of twin-width $2$ graphs.

In another direction, one may wonder if it is possible to compute other values of the twin-width efficiently for permutation graphs. An FPT approximation algorithm is already known in this setting \cite{twin-width1}, although one may try to find tighter approximation bounds.

\paragraph*{Organisation}
We organise this paper as follows.
In Section~\ref{sec:prelim}, we present some terminology and useful known results on twin-width, and on permutation graphs.
In Section~\ref{sec:permutation}, we show that every graph of twin-width at most~$1$ is a permutation graph.
In Section~\ref{sec:algo}, we present a linear-time algorithm that recognises graphs of twin-width at most~$1$.
In Section~\ref{sec:DH}, we present results related to distance-hereditary graphs and split decompositions.

\section{Preliminaries}\label{sec:prelim}

In this paper, all graphs are finite and simple.
For an integer~$i$, we denote by~$[i]$ the set of positive integers at most~$i$.
Note that if $i\leq0$, then~$[i]$ is an empty set.
A subset $I\subseteq[i]$ is an \emph{interval} of~$[i]$ if there are integers $i_1,i_2\in[i]$ such that $I=\{j:i_1\leq j\leq i_2\}$.

Let~$G$ be a graph. We denote the complement graph by $\overline{G}$.
For a vertex~$v$ of~$G$, we denote by $N_G(v)$ the set of neighbours of~$v$, and let $N_G[v]:=N_G(v)\cup\{v\}$. Then, $\overline{N}_G(v) = V(G) \setminus N[v]$.
For a set $X\subseteq V(G)$, let $N_G[X]:=\bigcup_{v\in X}N_G[v]$ and let $N_G(X):=N_G[X]\setminus X$.
Distinct vertices~$v$ and~$w$ of~$G$ are \emph{twins in~$G$} if $N_G(v)\setminus\{w\}=N_G(w)\setminus\{v\}$.
For a set $X\subseteq V(G)$, we denote by $G[X]$ the subgraph of~$G$ induced by~$X$.
A \emph{dominating set} of~$G$ is a set $D\subseteq V(G)$ such that $N_G[D]=V(G)$.
A \emph{dominating path} of~$G$ is a path whose vertex set is a dominating set of~$G$.
Three vertices form an \emph{asteroidal triple} if each pair of vertices is connected by a path that does not dominate the third vertex.
A graph is \emph{asteroidal triple-free} (AT-free) if it contains no asteroidal triple.
An \emph{independent set} of a graph~$G$ is a set $I\subseteq V(G)$ such that~$G$ has no edge between two vertices in~$I$.

A \emph{module} of a graph~$G$ is a set $M\subseteq V(G)$ such that for every vertex $v\in V(G)\setminus M$, $v$ is either adjacent to every vertex in~$M$, or nonadjacent to every vertex in~$M$. If $M$ is not a module, we call \emph{splitter} of $M$ a vertex $v \in V(G) \setminus M$, that has both a neighbour and a non-neighbour in $M$. Note that $M$ is a module exactly if it has no splitter.
A module is \emph{trivial} if it either has size at most~$1$, or is equal to~$V(G)$.
A graph is \emph{prime with respect to modules}, or simply \emph{prime}, if all of its modules are trivial.
A \emph{modular partition} of a graph~$G$ is a partition of~$V(G)$ where each part is a module of~$G$.
For a modular partition $\mathcal{M}=(M_1,\ldots,M_\ell)$ of a graph~$G$, we denote by~$G/\mathcal{M}$, the subgraph of~$G$ induced by a set containing exactly one vertex from each~$M_i$.
A module $M$ is \emph{strong} if it does not overlap with any other module, i.e. for every module $M'$, we have $M \subseteq M'$, or $M' \subseteq M$, or $M \cap M' = \varnothing$.

The \emph{modular decomposition} of a graph~$G$ is a tree $T$ corresponding to the recursive modular partition by the maximal strong modules. The leaves of the tree are vertices of $G$, and each internal node $n$ has an associated quotient graph $Q(n)$ which encodes the adjacency relation between vertices of $G$ introduced in subtrees below $n$. In particular, quotient graphs are induced subgraphs of $G$, and the modular decomposition of an induced subgraph of $G$ can easily be deduced from the modular decomposition of $G$. There are three types of internal nodes: \emph{series}, \emph{parallel}, and \emph{prime}, corresponding to cliques, independent sets and prime graphs (with respect to modules). Prime graphs are graphs whose modules are all trivial i.e. singletons and the complete vertex set. Series and parallel nodes are called \emph{degenerate}, and two degenerate nodes of the same type may not be adjacent in $T$. The reason they are called degenerate is because any subset of vertices of a clique or an independent set is a module.
Cographs are exactly the graphs whose modular decompositions have no prime nodes, their modular decomposition is usually called a cotree.

\subsection{Permutation graphs}\label{subsection:permutation-intro}

A graph~$G$ on~$n$ vertices is a \emph{permutation graph} if there are two linear orderings $\sigma:V(G)\to[n]$ and $\tau:V(G)\to[n]$ such that two vertices~$u$ and~$v$ are adjacent in~$G$ if and only if $\{u,v\}$ is an inversion of $\sigma^{-1}\circ\tau$.
We remark that $\sigma^{-1}\circ\tau$ and its inverse permutation $\tau^{-1}\circ\sigma$ have the same set of inversions.
A \emph{permutation diagram} of~$G$ with respect to~$\sigma$ and~$\tau$ is a drawing of~$n$ line segments between two parallel lines~$\ell_1$ and~$\ell_2$ such that each of~$\ell_1$ and~$\ell_2$ has~$n$ distinct points, and for each $v\in V(G)$, there is a line segment between $\sigma(v)$-th point of~$\ell_1$ and $\tau(v)$-th point of~$\ell_2$.
Note that two vertices are adjacent in~$G$ if and only if their corresponding line segments intersect each other in the permutation diagram.

We denote by $G[\sigma,\tau]$ the graph on $V$ where $uv$ is an edge if and only if $\{u,v\}$ is an inversion of $\sigma^{-1} \circ \tau$. As pointed earlier, we have $G[\sigma,\tau]=G[\tau,\sigma]$. If $G=G[\sigma,\tau]$, we call $(\sigma,\tau)$ a \emph{realiser} of $G$ (as a permutation graph).

For $\sigma$ an ordering of $V$, we define intervals of $V$ as follows: for $\{i,i+1,\dots,j-1,j\}$ an interval of $[n]$, we have an interval of $V$ for $\sigma$ $\{\sigma^{-1}(i)=u,\dots,\sigma^{-1}(j)=v\}$ that we denote by $[u,v]_{\sigma}$. We extend the notation to open intervals similarly.

We say that $I \subseteq V$ is an \emph{interval} of realiser $(\sigma,\tau)$ if it is an interval for $\sigma$ or for $\tau$. We also say that $u,v \in V(G_i)$ are \emph{consecutive} for $(\sigma,\tau)$ if $\{u,v\}$ is an interval of $(\sigma,\tau)$. Similarly, two intervals are \emph{consecutive} if their union is an interval. We say that $I$ is a \emph{common interval} of realiser $(\sigma,\tau)$ if it is an interval for $\sigma$ and for $\tau$. A vertex $v$ of $G$ is \emph{extremal} for realiser $(\sigma,\tau)$ (or equivalently the corresponding permutation diagram) if it is the first or last vertex of~$\sigma$ or~$\tau$.

It is known that for a prime permutation graph, its permutation diagram (or equivalently its realiser) is unique up to symmetry \cite{Gallai,Gol80}. More generally, the modular decomposition encodes all possible realisers. Furthermore, we have the following known observation which will be very convenient for our analysis.

\begin{observation}\label{obs:common-itv}
    If $M$ is a common interval of a realiser of $G$, then $M$ is a module of $G$. Conversely, if $M$ is a strong module of $G$, then it is a common interval of all realisers.
\end{observation}

It will also be useful to note that extremal vertices of a prime permutation graph do not depend on the realiser. See \cite{HabibP10,Montgolfier03,BergeronCMR08,CapelleHM02} for further details on this.

Moreover, for a permutation diagram of~$G$ with respect to linear orderings~$\sigma$ and~$\tau$ of~$V(G)$, if we reverse one of~$\sigma$ and~$\tau$, then we obtain a permutation diagram of~$\overline{G}$.
In other words, $\overline{G}$ has a permutation diagram with respect to~$\sigma$ and~$n+1-\tau$, or to~$n+1-\sigma$ and~$\tau$.

Thus, a graph~$G$ is a permutation graph if and only if $\overline{G}$ is a permutation graph.

\subsection{Trigraphs and twin-width}

A \emph{trigraph} is a triple $H=(V(H),B(H),R(H))$ where~$B(H)$ and~$R(H)$ are disjoint sets of unordered pairs of~$V(H)$.
We denote by~$E(H)$ the union of~$B(H)$ and~$R(H)$.
The \emph{edges} of~$H$ are the elements in~$E(H)$.
The \emph{black edges} of~$H$ are the elements of~$B(H)$, and the \emph{red edges} of~$H$ are the elements of~$R(H)$.
We identify a graph $G=(V,E)$ with a trigraph $(V,E,\emptyset)$.

The \emph{underlying graph} of a trigraph~$H$ is the graph $(V(H),E(H))$.
We identify~$H$ with its underlying graph when we use standard graph-theoretic terms and notations.
A \emph{black neighbour} of~$v$ is a vertex~$u$ with $uv\in B(H)$ and a \emph{red neighbour} of~$v$ is a vertex~$w$ with $vw\in R(H)$.
The \emph{red degree} of~$v$ is the number of red neighbours of~$v$.
For an integer $d$, a \emph{$d$-trigraph} is a trigraph with maximum red degree at most~$d$.
For a set $X\subseteq V(H)$, we denote by $H-X$ the trigraph obtained from~$H$ by removing all vertices in~$X$.
If $X=\{v\}$, then we write~$H-v$ for~$H-X$.

For a trigraph~$H$ and distinct vertices~$v$ and~$w$ of~$H$, we denote by $H/\{u,v\}$ the trigraph~$H'$ obtained from $H-\{u,v\}$ by adding a new vertex~$x$ such that for every vertex $y\in V(H)\setminus\{u,v\}$, the following hold:
\begin{itemize}
    \item if~$y$ is a common black neighbour of~$u$ and~$v$, then~$xy\in B(H')$,
    \item if~$y$ is adjacent to none of~$u$ and~$v$, then~$xy\notin E(H')$, and
    \item otherwise, $xy\in R(H')$.
\end{itemize}
We say that~$H'$ is obtained by \emph{contracting $u$ and $v$}.

A \emph{contraction sequence} of~$G$ is a sequence of trigraphs $G_n,\dots,G_1$ where $G_n=G$, $G_1$ as a single vertex, and $G_{i-1}$ is obtained from $G_i$ by contracting a pair of vertices.
For an integer $d$, a contraction sequence $G_n,\ldots,G_1$ is a \emph{$d$-sequence} if for every $i\in[t]$, the maximum red-degree of~$G_i$ is at most~$d$.
The \emph{twin-width} of $G$, denoted by $\tww(G)$, is the minimum integer~$d$ such that there is a $d$-sequence of~$G$.

The vertices of the trigraphs in the contraction sequence can be interpreted as a partition of $V(G)$, this leads to a very natural characterisation of red edges in a trigraph as pairs of part such that the relation between vertices in the two parts is not homogeneous (i.e. there is an edge and a non-edge).

We have the following observations from~\cite{twin-width1}.

\begin{observation}[Bonnet, Kim, Thomass\'{e}, and Watrigant~\cite{twin-width1}]\label{obs:subgraph-tww}
    For a graph~$G$, the twin-width of~$G$ is equal to that of~$\overline{G}$, and every induced subgraph of~$G$ has twin-width at most~$\tww(G)$.
\end{observation}

\begin{lemma}[Bonnet, Kim, Reinald, Thomass\'{e}, and Watrigant~{\cite[Lemma~9]{BonnetKRTW22}}]\label{lem:seq-modular}
    Let~$G$ be a graph and let~$\mathcal{M}=(M_1,\ldots,M_\ell)$ be a modular partition.
    Then
    \[
        \tww(G)=\max\left\{\tww(G/\mathcal{M}),\max_{i\in[\ell]}\tww(G[M_i])\right\}.
    \]
\end{lemma}

We sketch the proof as this will be of importance throughout the paper.

\begin{proof}
    A contraction in a module $M$ does not create red edges incident to $V(G-M)$. Therefore, we can always first contract the induced subgraphs $G[M_i]$ in any order, and then contract $G/\mathcal{M}$ once all subgraphs $G[M_i]$ have been contracted to single vertices.
\end{proof}

\begin{corollary}\label{coro:seq-modular}
    The twin-width of a graph $G$ is equal to the maximum twin-width of the quotient graphs of nodes of its modular decomposition.
\end{corollary}

\begin{proof}
    We may always contract leaf nodes of the modular decomposition via their optimal contraction sequence until the whole graph is reduced to a single vertex. This corresponds exactly to a recursive application of the above lemma with maximal strong modules.
\end{proof}

The above statements show that computing the twin-width of a graph reduces to the case of a prime graph. Thus, throughout this paper, we focus on prime graphs of twin-width at most~$1$.
The following lemma on the structure of twin-width $1$ graphs will be useful.

\begin{lemma}[Bonnet, Kim, Reinald, Thomass\'{e}, and Watrigant~{\cite[Lemma~31]{BonnetKRTW22}}]\label{lem:seq-prime}
    Let $G$ be a prime graph of twin-width~$1$ and let $G_n(=G),\ldots,G_1$ be a $1$-sequence of~$G$.
    Then for every $i\in[n-1]\setminus\{1\}$, the trigraph~$G_i$ has exactly one red edge.
\end{lemma}

\section{Permutation diagrams}\label{sec:permutation}

In this section, we show that every graph of twin-width at most~$1$ is a permutation graph.

\begin{theorem}\label{thm:permutation}
    Every graph of twin-width at most~$1$ is a permutation graph.
\end{theorem}

It would be sufficient to prove that graphs of twin-width at most 1 are comparability graphs.
Indeed, since twin-width is stable by complementation, this implies that graphs of twin-width at most 1 are also co-comparability graphs.
We could then conclude because permutation graphs are exactly the graphs that are both comparability and co-comparability graphs \cite{DushnikMiller}.
To do this, one could make use of the following characterisation of comparability graphs via forbidden induced subgraphs given by Gallai \cite{Gallai} (see also \cite{thesisATfree}).

\begin{theorem}
    A graph $G$ is a comparability graph if and only if $G$ does not contain any graph of \cref{fig:3asteroid} as induced subgraph, and $\overline{G}$ does not contain any graph of \cref{fig:largeasteroid} as induced subgraph.
\end{theorem}

We instead present a constructive proof because it gives combinatorial insights on the relationship between contraction sequences and permutation diagrams that will be useful to find an efficient recognition algorithm. The trick is to decompose the graph by considering the contraction sequence backwards and to realise that vertices that are contracted together at some point in the contraction sequence should be consecutive or almost consecutive in a realiser.

The following recursive decomposition of twin-width 1 graphs is meant as a simple introduction to the technique.
This decomposition will later be strengthened in Lemma~\ref{lem:build-realiser} and Corollary~\ref{coro:respecting-realiser}.
We abusively extend the terminology of universal vertex and isolated vertex to modules $M$ such that replacing $M$ by a single vertex makes it universal or isolated.

\begin{lemma}\label{lem:gen-tww1}
    The class $\mathcal{G}$ of graphs of twin-width at most 1 satisfies the following recursive characterisation: $$\mathcal{G} = \left\{G ~|~ \exists M \subset V(G), \begin{aligned} ~ & G - M \in \mathcal{G}, G[M] \in \mathcal{G}, M \text{ is a module of $G$}, \\ ~ & M\text{ is universal or } M \text{ is isolated or}\\ ~ &  G - M \text{ admits a 1-contraction sequence ending}\\ ~ & \text{with the (possibly red) edge } \{N(M),\overline{N}(M)\} \end{aligned}\right\}.$$
\end{lemma}

\begin{proof}
    Consider a graph $G$ such that there exists a module $M \subset V(G)$ such that $G - M \in \mathcal{G}$ and $G[M]\in \mathcal{G}$. If $M$ is universal or isolated, $G-M$ is a module of $G$ and, by \cref{lem:seq-modular}, $G$ is also a graph of $\mathcal{G}$. If $G-M$ admits a partial $1$-contraction sequence $\pi$ ending with edge $\{N(M),\overline{N}(M)\}$, then $\pi$ is also a partial $1$-contraction sequence of $G$, and can be extended to be $1$-contraction sequence of $G$. Indeed, no red edge incident to $M$ appears when applying~$\pi$ to~$G$ since pairs of contracted vertices are always either in $N(M)$ or $\overline{N}(M)$, we can then apply the $1$-contraction sequence of $M$, and finally contract arbitrarily the resulting 3-vertex trigraph. 

    Conversely, consider a graph of $\mathcal{G}$. It admits a $1$-contraction sequence with a $3$-vertex trigraph $G_3$. $G_3$ has at most one red edge $ab$ (pick $a,b$ arbitrarily in $V(G_3)$ if there is no red edge). We consider $v \in V(G_3)-\{a,b\}$. The set $M$ of vertices of $G$ that were contracted to form $v$ is a module since there are no incident red edges in $G_3$. $G-M\in \mathcal{G}$ and $G[M]\in \mathcal{G}$ since they are induced subgraphs of $G$. Let $A$ and $B$ be the sets of vertices that were contracted to form $a$ and $b$, respectively. Since $v$ has no red neighbour, it must be that $A \subseteq N(M)$ or $A \subseteq \overline{N}(M)$, similarly, $B \subseteq N(M)$ or $B \subseteq \overline{N}(M)$. We conclude that $M$ is universal, isolated or $G-M$ admits a contraction sequence ending with the edge $\{N(M),\overline{N}(M)\}$.
\end{proof}

Given a permutation graph $G=(V,E)$ and $U \subseteq V$, we denote by $\sigma[U]$ and $\tau[U]$ the restrictions of $\sigma$ and $\tau$ (re-numbered by $[|U|]$). They form a realiser of $G[U]$ as a permutation graph. We also use the notation $\sigma \cdot \tau$ to denote the concatenation of two orderings on disjoint domains (all elements of $\sigma$ are before elements of $\tau$, relative orders in $\sigma$ and $\tau$ are preserved).

The following lemma implies \cref{thm:permutation}.
\begin{lemma}\label{lem:build-realiser}
    Let $G$ be a graph of twin-width at most $1$, and $G_n,\dots G_2,G_1$ be a fixed $1$-contraction sequence of $G$.
    
    Then $G$ is a permutation graph admitting a realiser $(\sigma,\tau)$ satisfying the following.
    \begin{itemize}
        \item For every $x\in V(G_i)$, $\{v\in V(G) | v \in x\}$ is an interval of $(\sigma,\tau)$.
        \item For every $x,y \in V(G_i)$, if $xy\in R(G_i)$, then $\{v\in V(G) | v \in (x \cup y) \}$ is an interval on one of the orderings $\pi \in \{\sigma,\tau\}$, and both $\{v\in V(G) | v \in x\}$ and $\{v\in V(G) | v \in y\}$ are intervals on the other ordering. 
    \end{itemize}
\end{lemma}
\begin{proof}
    We proceed by induction on the number of vertices of the graph by leveraging the recursive characterization given in \cref{lem:gen-tww1}.

    The statement is trivial for graphs on at most $3$ vertices.
    We now consider a graph $G \in \mathcal{G}$ on more than $3$ vertices and its contraction sequence $G_n,\dots,G_3,G_2,G_1$. Let $M$ be the set of vertices of $G$ that were contracted into a vertex of $G_3$ not incident to a red edge in $G_3$, $M$ is a module of $G$. 
    Restricting the sequence $G_n,\dots,G_1$ to $G - M$ and $G[M]$, respectively, yields a $1$-contraction sequence for each of them.
    By induction hypothesis applied with these sequences, we obtain realisers $\sigma_1,\tau_1$ and $\sigma_2,\tau_2$ of $G-M$ and $G[M]$, respectively. The following cases are illustrated in \cref{fig:tww1-realiser}.
    If $M$ is universal (resp. isolated), $\sigma_1 \cdot \sigma_2,\tau_2 \cdot \tau_1$ (resp. $\sigma_1 \cdot \sigma_2,\tau_1 \cdot \tau_2$) is a realiser of $G$. 
    Otherwise, the contraction sequence $G_n-M,\dots,G_3-M$ is such that $V(G_3-M)=\{N(M),\overline{N}(M)\}$. In particular, we know from the induction hypothesis that $B=N(M)$ and $A = \overline{N}(M)$ are intervals of either $\sigma_1$ or $\tau_1$, as they correspond to the sets of vertices contracted into the last two vertices on the sequence. 
    Without loss of generality, we assume they are intervals of $\sigma_1$. We can now produce the realiser $\sigma_1[B] \cdot \sigma_2 \cdot \sigma_1[A],\tau_2 \cdot \tau_1$.
    In all cases, the realiser satisfies the desired properties: it suffices to observe that by construction we do not contract vertices from $G-M$ with vertices from $G[M]$ until the last two contractions, so we must check only the last two contractions which are on trigraphs of order at most $3$ hence the properties are satisfied. 
\end{proof}

\begin{figure}[H]
    \centering
    \includegraphics[scale=0.8]{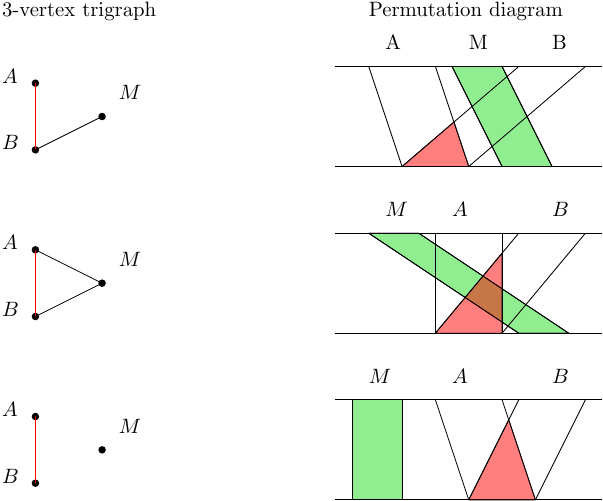}
    \caption{The different cases of the induction constructing the realiser.}
    \label{fig:tww1-realiser}
\end{figure}

We now move on to show related properties that will be useful in further understanding permutation diagrams and their relation to contraction sequences.

In the particular case of a prime graph, we deduce the following two corollaries of \cref{lem:build-realiser}.

\begin{corollary}\label{cor:first-contraction}
    In a prime graph $G$ of twin-width $1$, the first contraction must involve two vertices that are consecutive for one ordering, and with exactly one vertex between them for the other ordering.
\end{corollary}

\begin{proof}
    In $G_{n-1}$, let $x',z$ be the vertices incident to the red edge (there is a red edge because the graph is prime, see \cref{lem:seq-prime}) with $x'$ resulting of the contraction of vertices $x$ and $y$ of~$G$. $X'=\{x,y\}$ is an interval for one ordering $\sigma$ of the realiser. $z$ must be a splitter of $X'$ because it is incident to the red edge, this implies that $z$ is between $x$ and $y$ in the ordering~$\tau$.
\end{proof}

Unfortunately, there can be linearly many such pairs. On the other hand, there are only~$4$ extremal vertices in a prime permutation graph. Hence, the following corollary is useful for obtaining a linear time recognition algorithm.

\begin{corollary}\label{cor:last-vertex}
    In a prime graph of twin-width $1$, for any $1$-contraction sequence, the last vertex to become incident to the red edge is an extremal vertex of the realiser. 
\end{corollary}

\begin{proof}
    Due to the graph being prime, the last vertex of $G$ to become incident to the red edge is also a vertex of $G_3$ and is adjacent to exactly one other vertex of $G_3$. The other two vertices of $G_3$ are connected because prime graphs are connected, hence $G_3$ is isomorphic to~$P_3$ whose vertices are all extremal. In particular, the last vertex incident to the red edge is extremal.
\end{proof}

\begin{observation}\label{obs:trigraphs-induced}
    If $G_n,\dots,G_1$ are the trigraphs of a 1-contraction sequence of $G$, their underlying graphs form a chain for the induced subgraph relation, i.e. $G_{i-1}$ is an induced subgraph of $G_i$ for $i \in [2,n]$.
\end{observation}

\begin{proof}
    We describe how to construct the induced subgraph of $G$ by picking a representative in $G$ for each vertex in $G_i$.
    
    Consider a vertex $v$ of $G_i$, either it has only black incident edges and we may pick any vertex of $G$ that was contracted into $v$, or it has exactly one incident red edge $vw$, in which case we can pick the endpoints of some edge of $G$ that is between the vertices contracted into $v$ and the vertices contracted into $w$.
    
    This construction picks only one vertex of $G$ per vertex of $G_i$ because the red degree is at most one meaning we do not pick more than once.
\end{proof}

This property only holds for $1$-contraction sequences (e.g. contracting non-adjacent vertices of a matching creates a path on $3$ vertices which is not an induced subgraph). It allows to view the contraction sequence as a sequence of vertex deletions. In particular, underlying graphs of the trigraphs of a $1$-contraction sequence are also permutation graphs and one of their permutation diagrams is the induced permutation diagram obtained by seeing $G_i$ as an induced subgraph of $G$. Observe that this diagram preserves the relative order of vertices that are not deleted. This will be important for our inductive reasoning since we can keep the same diagram while decomposing the graph.
We denote by $\sigma[G_i]$ the ordering induced by $G_i$ seen as an induced subgraph of $G$.

\begin{corollary}\label{coro:respecting-realiser}
    Let $G$ be a graph and $G_n,\dots, G_1$ a $1$-contraction sequence of $G$.
    The realiser constructed from $G_n,\dots,G_1$ using Lemma~\ref{lem:build-realiser} has the property that, for every $i>1$, the vertices being contracted from $G_i$ to $G_{i-1}$ are consecutive for $(\sigma[G_i],\tau[G_i])$, and endpoints of red edges are consecutive for $(\sigma[G_i],\tau[G_i])$.
    
    Conversely, for any realiser $(\sigma',\tau')$, there exists a $1$-contraction sequence $G_n,\dots,G_1$ such that for every $i>1$ the contracted vertices of $G_i$ are consecutive in $(\sigma'[G_i],\tau'[G_i])$ and for any $xy \in R(G_i)$, $x$ and $y$ are consecutive in $(\sigma'[G_i],\tau'[G_i])$.
\end{corollary}

\begin{proof}
    Suppose we contract $a,b \in V(G_i)$ into vertex $c \in V(G_{i-1})$. Let $C$ be the set of vertices of $G$ contracted into $c$, and $A,B$ those contracted into $a,b$. By Lemma~\ref{lem:build-realiser}, $A$, $B$ and $C$ are all intervals of $(\sigma,\tau)$. Since $C = A \cup B$, we may assume without loss of generality they are all intervals of $\sigma$. Hence, $a$ and $b$ are consecutive for $\sigma[G_i]$.
    Similarly, the set of vertices of $G$ contracted into a fixed red edge is an interval of $(\sigma,\tau)$. Since the sets of vertices of $G$ incident to the red edge are also intervals, this means the endpoints of the red edge in $G_i$ are consecutive.

    Now, consider any realiser $(\sigma',\tau')$ for graph $G$, and let us show the second part of the lemma by induction on the modular decomposition of $G$.
    The result holds trivially for leaves of the decomposition. Let us then consider an internal node $q$ of the decomposition of $G$ and consider the subgraphs $H_1,...,H_k$ of $G$ induced by the children of $q$. Now, since each $H_i$ is a strong module of $G$, they each form a common interval of $(\sigma',\tau')$ by Observation~\ref{obs:common-itv}. They also have twin-width $1$, so our induction hypothesis along with \cref{lem:seq-modular} yields that we can contract each child fully by respecting the consecutivity properties on $(\sigma'[H_i],\tau'[H_i])$, and thus on $(\sigma',\tau')$.
    Then, if $q$ is a degenerate node, we can always find a pair of consecutive twins and contract them. Otherwise, $q$ is prime, and as such it admits a unique realiser (up to symmetry), which is given by \cref{lem:build-realiser}, for which the contraction sequence satisfies the desired properties by the first part of the lemma.
\end{proof}

\section{Linear-time recognition algorithm}\label{sec:algo}

A realiser $(\sigma,\tau)$, if it exists, can be found in linear time \cite{McConnellS99}.
From the Corollaries \ref{coro:seq-modular}, \ref{cor:first-contraction} and \ref{coro:respecting-realiser}, we deduce that the algorithm of \cite{BonnetKRTW22} can be implemented in time $O(n^2 + m)$. Indeed, for every prime node of the modular decomposition, it suffices to compute $\sigma$ and $\tau$ for its quotient graph.
At this point, we have restricted the set of possible first contractions to the set all pairs of consecutive vertices, and we may then simulate contraction sequences efficiently using the fact that new vertices to be contracted can be found in constant time, since they are consecutive. In order to obtain a linear time algorithm, we will instead guess the end of the sequence, and try to extend the sequence from its end, see Figure~\ref{fig:dec-illustration}.

\begin{lemma}\label{lem:dec-tww1}
    Let $G$ be a prime permutation graph of twin-width 1, which admits a $1$-contraction sequence where an extremal vertex $s$ is last to become incident to a red edge.
    Let also $G^1=G[N(s)]$ and $G^2=G[\overline{N}(s)]$.
    
    Then, there is a module $M$ of $G^i$ for some $i$ that has a unique splitter $s' \in V(G^{3-i})$, such that $G[M \cup \{s'\}]$ admits a $1$-contraction sequence where $s'$ is the last vertex to be incident to the red edge, where:
    \begin{enumerate}[(i)]
        \item either $M$ is a pair of consecutive twins in $G^i$ and there are no prime nodes in the modular decompositions of $G^1$ and $G^2$,
        \item or $M$ corresponds to the subtree rooted at a prime node $p$ in the modular decomposition, and $p$ is the unique prime node in the modular decompositions of $G^1$ and $G^2$.
    \end{enumerate}
    Moreover, let $k=|V(G) - (M \cup \{s'\})|,$ there is an ordering $\pi:[k] \to V(G) - (M \cup \{s'\})$ such that for all $i \in [k]$, $M \cup \{s'\} \cup \pi([i])$ forms an interval of $\sigma$ and two intervals of $\tau$ for a realiser $(\sigma,\tau)$ where $\sigma$ is the linear ordering for which $s$ is extremal.
\end{lemma}

\begin{proof}
    Let $(\sigma,\tau)$ be a realiser of $G$ and $s$ be an extremal vertex for $\sigma$ which is the last vertex incident to the red edge in some $1$-contraction sequence of $G$. Such a vertex exists due to Corollary~\ref{cor:last-vertex}. Let $G^1=G[N(s)]$ and $G^2=G[\overline{N}(s)]$ and 
    observe that since $G$ is prime, $G^1$ and $G^2$ are not empty.  As $s$ is extremal for $\sigma$, adjacent to $V(G^1)$ and non-adjacent to $V(G^2)$, $V(G^1)$ and $V(G^2)$ must be disjoint intervals of $\tau$.

\begin{claim}\label{claim:seq-partition}
    In any contraction sequence as considered, we only contract pairs of vertices of $G^1$ or pairs of vertices of $G^2$ until we reach the $3$-vertex trigraph.
\end{claim}

\begin{subproof}
    Due to the assumption on $s$ being the last vertex incident to a red edge, we cannot contract a vertex of $G^1$ with a vertex of $G^2$, unless $G^1$ and $G^2$ both have been contracted to a single vertex, because they are split by $s$.
\end{subproof}

\begin{claim}\label{claim:attribute-splitter}
    For any non-trivial module $M$ of $G^i$ there is a splitter $s' \in V(G^{3-i})$. Furthermore, for any splitter $s' \in V(G^{3-i})$ of $M$ there are $a,b$ in $M$ such that $\sigma(a) < \sigma(s') < \sigma(b)$.
\end{claim}
    
\begin{subproof}
    Indeed, $M$ cannot be a module of $G$, as $G$ is a prime graph, so it must have a splitter $s'$ in $G-V(G^i)$. Also note that $s$ is not a splitter of $M$ by definition of $G^i$. Furthermore, since $s'$ is a splitter it must be adjacent to some $a \in M$ and non-adjacent to some $b \in M$. Since $V(G^1)$ and $V(G^2)$ are intervals of $\tau$, we get that $\sigma(a) < \sigma(s') < \sigma(b)$ or $\sigma(a) < \sigma(s') < \sigma(b)$ by definition of the permutation diagram.
\end{subproof}

The following observation will allow to simplify the analysis.

\begin{observation}\label{obs:prime-extract}
    If $X$ induces a prime subgraph of $H$, and $t \in V(H) \setminus X$ is a splitter of~$X$, then we can both extract a pair of vertices of~$X$ that are adjacent and split by $t$, and a pair of vertices of $X$ that are non-adjacent and split by $t$.
\end{observation}

\begin{proof}
    Consider the cut $C=(X\cap N(t),X - N(t))$ in $H[X]$, since $H[X]$ is prime both $H[X]$ and $\overline{H[X]}$ are connected, yielding both an edge and an non-edge across $C$.
\end{proof}

\begin{claim}\label{claim:strong-module-unique}
    There cannot be disjoint non-trivial strong modules in $G^i$.
\end{claim}

\begin{subproof}
    We proceed by contradiction and consider $M,M'$ two maximal disjoint non-trivial strong modules of $G^i$, without loss of generality $i=1$. We first show that, in $G$, $M$ and~$M'$ must have distinct splitters from $V(G^2)$, and then conclude that this contradicts the assumption that $G$ has a $1$-contraction sequence respecting \cref{claim:seq-partition}. 
    
    Using the fact that $V(G^1)$ and $V(G^2)$ are intervals of $\tau$ and the common interval characterisation of strong modules of $G^1$ given by \cref{obs:common-itv}, we deduce that a splitter~$t$, which necessarily belongs to $V(G^{2})$, cannot split both $M$ and $M'$. Indeed, we have $\sigma(M) < \sigma(M')$ or $\sigma(M') < \sigma(M)$, so if $t$ is a splitter of $M$ there exist $a,b \in M$ such that $\sigma(a) < \sigma(t) < \sigma(b)$ as seen in the previous claim. In particular, $t$ is not a splitter of $M'$.

    Now because $M$ and $M'$ are non-trivial modules of $G^1$, they each have a splitter in $G^{2}$. Let $t$ and $t'$ be such splitters. By maximality, $M$ and $M'$ correspond to subtrees rooted at siblings of the same node $n$ in the modular decomposition of $G^1$. 
    Let us first consider the case when $n$ is a prime node. In this case, Lemma~\ref{lem:seq-prime} guarantees a red edge for $Q(n)$, and the two splitters $t,t'$ forbid the existence of a contraction sequence with consecutivity properties as guaranteed by Lemma~\ref{lem:build-realiser} and Corollary~\ref{coro:respecting-realiser}.

    We may now assume that $n$ is a degenerate node, and exhibit an induced subgraph $H$ of $G$ that does not admit a $1$-contraction sequence respecting our assumptions. This will contradict the existence of such a sequence for the whole $G$, as its restriction to $H$ would also be valid (as in Observation~\ref{obs:subgraph-tww}).
    The graph $H$ consists of splitters $t,t'$, as well as two pairs of twins $a,b \in M$ and $a',b' \in M'$ distinguished by $t,t'$ respectively. Then, vertices $a,b,a',b'$ induce a cograph whose corresponding cotree is a complete binary tree of depth $2$.
    In other words, $(a,b)$ and $(a',b')$ are present if and only if $\{a,b\}$ is adjacent to $\{a',b'\}$.

    We first reduce $M$ and $M'$ so as to induce exactly the quotient graph of their corresponding node in the modular decomposition (recall that this is well defined because they are strong) by picking one vertex in each subtree below this node. We may ensure that $t$ and $t'$ still split this subset of vertices (e.g. by picking extremal vertices of the interval of $\sigma$ corresponding to~$M$ and~$M'$, recalling \cref{obs:common-itv}).

    Now we observe that if $M$ (resp. $M'$) has a degenerate root node, it is of the opposite type of the node $n$, and any pair of vertices of $M$ (resp. $M'$) that is split by $t$ (resp. $t'$) can be taken for our subgraph $H$.
    Otherwise, $M$ (resp. $M'$) has a prime root node, and we apply Observation~\ref{obs:prime-extract} to obtain a pair of vertices that is split and has the required adjacency relation.
    
    Let us now show that $H$ does not admit a $1$-contraction sequence satisfying our assumptions. Indeed, $t$ and $t'$ may only be contracted together, because they are the only vertices in $V(G^{2})$ of our subgraph (Claim~\ref{claim:seq-partition}), but they are split by at least two vertices in $V(G^1)$. Now, any order of contractions on $a,b,a',b'$ will create two red edges. More precisely, contracting a vertex from each nontrivial module creates two red edges, unless one module was already reduced to a single vertex, in which case it has a red edge to its splitter from $V(G^{2})$. This will inevitably contradict Lemma~\ref{lem:seq-prime}.
\end{subproof}

\begin{claim}\label{claim:prime-core}
    If $G^i$ has a prime node $p$ in its modular decomposition, the first red edge must appear in the module $M$ corresponding to the subtree rooted at $p$.
    In this case, $M$ has exactly one splitter $s'$ from $G^{3-i}$ and $s'$ should not be incident to a red edge until $M$ is contracted to a single vertex.
\end{claim}

\begin{subproof}
    Let $p$ be a prime node of the modular decomposition of $G^i$ and $M$ the module corresponding to the subtree rooted at $p$. Without loss of generality assume $i=1$.
    We first make the observation that, since there cannot be two disjoint nontrivial strong modules in~$G^1$ by \cref{claim:strong-module-unique}, and since prime graphs have at least $4$ vertices, there are vertices introduced by node~$p$.

    We first consider the case when a red edge incident to a vertex of $G^{2}$ is created before we fully contract $M$. This red edge has to remain incident to a vertex of $G^{2}$ due to Claim~\ref{claim:seq-partition}. Contracting a vertex introduced in prime node $p$ will create another red edge with both endpoints in $G^1$, which is thus different and contradicts Lemma~\ref{lem:seq-prime}.
    
    We now deal with the case where no red edge incident to a vertex of $G^{2}$ appears before we fully contract $M$. Assume the first red edge does not have both endpoints in $M$. Observe that the first contraction cannot be between vertices that are both not in $M$ but have a different adjacency to it because $M$ is a module on at least $4$ vertices, and this would create a red edge to all of them.
    Observe also that if the first contraction had involved a vertex of~$M$, there would also be a red edge with both endpoints in $M$. Indeed, on the one hand, if both contracted vertices were in $M$, red edges with both endpoints in $G^i$ would have both endpoints in $M$ because it is a module. On the other hand, $M$ is also a strong module with a prime node at its root, so all of its vertices have a mixed adjacency to the rest of $M$, contrary to vertices outside $M$. We conclude that the first contraction is between two vertices outside of $M$, so the corresponding red edge has both endpoints outside of $M$.
    This has the following implications: both contracted vertices have an homogeneous adjacency to $M$, but they also have a unique splitter in $G^1$, this implies we cannot contract the red edge before contracting the subgraph corresponding to $p$. Indeed, the two vertices incident to the red edge must have opposite adjacency relations with respect to $M$. Since contracting $M$ will force the creation of a second red edge, this contradicts Lemma~\ref{lem:seq-prime}.

    $M$ is a non-trivial strong module of $G^1$ so it has at least one splitter in $G^{2}$.
    From Lemma~\ref{lem:seq-prime} and Claim~\ref{claim:seq-partition}, it follows that a red edge incident to this splitter cannot appear before $M$ is contracted to a single vertex. Finally, only one red edge may appear when $M$ is contracted to a single vertex meaning the splitter must be unique.
\end{subproof}

In particular, this shows that only one of $G^1$ and $G^2$ can have a prime node in their modular decomposition.
We can moreover deduce that if there is a prime node in $G^i$, it must be unique. Indeed, a red edge has to appear in some $M$ corresponding to the subtree of the prime node that is deepest in the modular decomposition of $G^i$. At any step of the sequence, $M$ contains a red edge, until it is a single vertex with a red edge towards its splitter. Then, vertices introduced in any higher prime node cannot be contracted on such red edges without producing an additional red edge, which contradicts Lemma~\ref{lem:seq-prime}.

We are now ready to define $M$.
Consider the first red edge appearing between $X_1 \subseteq V(G^1)$ and $X_2 \subseteq V(G^2)$ in a $1$-contraction sequence as considered, and assume without loss of generality that it stems from a contraction in $G^1$.
Then, $X_2$ can only consist in a single vertex $s'$ due to Claim~\ref{claim:seq-partition}, Lemma~\ref{lem:seq-prime}, and the fact that this is the first red edge with an endpoint in both $G^1$ and $G^2$.
Then, $X_1$ is included in a module of $G^1$ which admits as unique splitter $s'$. 
In the case when there is no prime node, $X_1$ induces a cograph, so we may then take $M$ to be any pair of consecutive twins $a,b \in X_1$, which are also consecutive twins in $G^1$ by Corollary~\ref{coro:respecting-realiser}, distinguished by $s'$ (because $G$ is prime and $s'$ is the only splitter of $X_1$).

Otherwise, we can take $M$ and $s'$ as guaranteed by \cref{claim:prime-core} on $G^1$ and $G^2$.
The restriction of our contraction sequence to $G[M \cup s']$ yields a $1$-contraction sequence in which $s'$ becomes incident to a red edge last, as desired.

Having defined $M$, we are ready to show the following.
\begin{claim}\label{claim:reorganise}
    There exists a $1$-contraction sequence of $G$ that starts by contracting only vertices of $M$ until it is reduced to a single vertex. In particular, this creates a red edge towards $s'$ at the last contraction.
\end{claim}
\begin{subproof}
    We can extend the contraction sequence that contracts $M$ first into a sequence contracting the smallest interval $I$ of $\sigma[V(G^1)]$ containing $X_1$ because $X_1$ is not split by any splitter from $G^2$ other than $s'$, and it does not have disjoint non trivial strong modules. After contraction of $M$ to a single vertex, the remaining vertices of $I$ induce a cograph that is not split by any vertex of $G^2$ other than $s'$. We can then proceed with the rest of the contraction of $G$ simply by contracting remaining vertices according to our initial contraction sequence.
\end{subproof}

At this point, in both cases, we know $M$ forms a common interval for $(\sigma[V(G^1)],\tau[V(G^1)])$, and that there is some $1$-sequence for $G$ that first contracts $M$ and in which $s$ is the last vertex to be incident to a red edge. We deduce that $M \cup \{s'\}$ satisfies the setting of Corollary~\ref{coro:respecting-realiser}, and any further contractions must be consecutive to $M$ due to Corollary~\ref{coro:respecting-realiser}, ensuring the existence of $\pi$.
\end{proof}

\begin{figure}[h]
    \centering
     \includegraphics[width=0.5\linewidth]{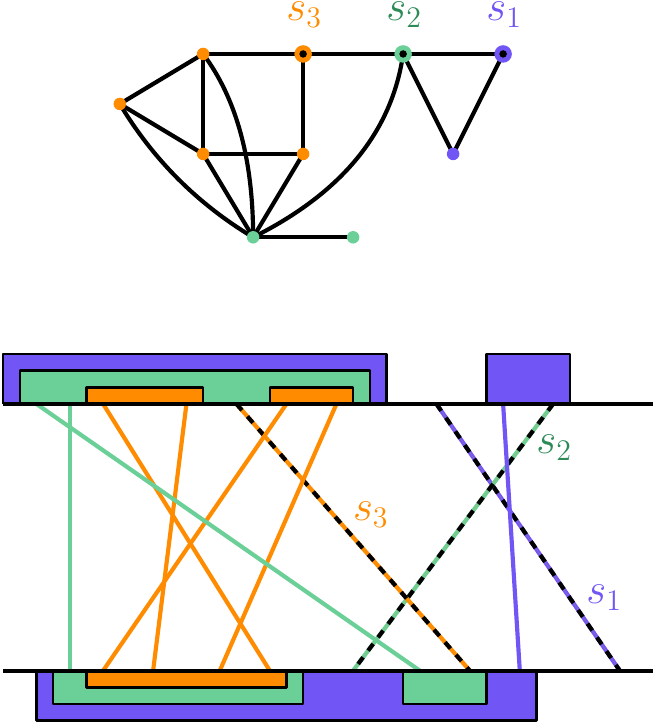}
    \caption{An example of a twin-width 1 graph and its permutation diagram. Shown here is the decomposition of Lemma~\ref{lem:dec-algo} starting from extremal vertex $s_1$. The segments of splitters $s_1,s_2,s_3$ are represented in dashed blue, green and orange respectively.
Then, intervals of the corresponding colour are the intervals $A,B,C$ produced during the iteration in which the splitter is considered.
$A$ and $B$ induce subgraphs $G^i$ from Lemma~\ref{lem:dec-tww1}.
Then, segments of the same colour as a splitter correspond to doubly extremal vertices eliminated in the corresponding step.}
    \label{fig:dec-illustration}
\end{figure}

The crucial observation now is that $\pi$ can easily be computed greedily using the realiser.

\begin{lemma}\label{lem:dec-algo}
    Given a permutation diagram of a prime permutation graph $G$ and an extremal vertex $s$ of $G$, we can produce in time linear in the number of vertices of $G$ a $1$-contraction sequence such that $s$ is the last vertex of $G$ to be incident to the red edge, or conclude that no such sequence exists.
\end{lemma}

\begin{proof}
    This is essentially an implementation of the result of \cref{lem:dec-tww1}, we reuse similar notations and provide an illustration in Figure~\ref{fig:dec-illustration}. The algorithm works as follows.

    We initialise indices at the endpoints of the intervals $A$ and $B$ of $\tau$ that are split by $s$, and at the endpoints of the interval $C=V(G)-\{s\}$ of $\sigma$. We add $s$ to the global list $\Pi$ and mark it as a splitter.

    While there is one endpoint $v$ of $C$ that is also an endpoint of $A$ or an endpoint of $B$, we add $v$ to the global list $\Pi$ and update indices describing $A,B,$ and $C$ in accordance to the removal of $v$. We call such a vertex $v$ \emph{doubly extremal}.

    If the intervals have become empty, return $\Pi$, otherwise, if either $A$ or $B$ contains a single vertex, we set $s'$ to be this vertex and $M$ to be the other interval, and apply the algorithm recursively to $G[M \cup \{s'\}]$ with $s'$ as the last vertex incident to a red edge. Finally, if both $A$ and $B$ contain more than one vertex, we reject.

    We can use marked vertices to deduce a contraction sequence: we iterate through $\Pi$, starting from the last vertices added.
    When reaching a marked vertex, we contract the red edge.
    Otherwise, we contract the current vertex to the vertex of the red edge with the same adjacency to the next marked vertex.

    If the algorithm does not reject, let us show by induction on the recursive calls that the list $\Pi$, along with its marked vertices, describes a valid order of contractions for $G$.
    First, we clarify the starting red edge that we consider in the part of the contraction sequence built at this step of the recursion.
    In the case when we emptied $A$ and $B$ (so $M$ and $s'$ are not defined), take the last vertex in $A$ and the last vertex in $B$ as a starting red edge.
    We move to the case where $M$ and $s'$ are defined.
    Here, by induction hypothesis, the algorithm finds a $1$-contraction sequence with $s'$ as the last vertex incident to a red edge for $G[M\cup \{s'\}]$.
    This contraction sequence for $G[M\cup \{s'\}]$ gives a partial $1$-contraction sequence for $G$ (see Observation~\ref{obs:subgraph-tww}) yielding a red edge between $M$ and $s'$, with all other vertices of $G$ not contracted yet.

    We now check that each vertex added to $\Pi$ at this step can be contracted in the order described by $\Pi$. This is the case because the vertex does not split the vertices of~$A$ contracted before it, nor the vertices of $B$ contracted before it. By contracting it to the vertex of the red edge with the same adjacency to $s$, we create no additional red edge. In particular, $s$ is still not incident to a red edge.
    Once we contracted all vertices of $A$ and $B$, the red edge can be safely contracted. Indeed, $s$ is the unique splitter of $A \cup B$.

    If the algorithm rejects, it contradicts the existence of $\pi$ guaranteed by \cref{lem:dec-tww1}.
    Indeed, if for all $i$, the $M \cup \{s'\} \cup \pi([i])$ are intervals as claimed, it must be that for each $i$, $\pi(i)$ is extremal on $\sigma[M \cup \{s'\} \cup \pi([i])]$ and on $\tau[(M \cup \{s'\} \cup \pi([i]))\cap V(G^j)]$ for some $j \in \{1,2\}$. In particular, all vertices that are not in $M \cup \{s'\}$ will become doubly extremal if $\pi$ exists. Note that this is independent of the order in which they are contracted. Indeed, a vertex remains (doubly) extremal while we delete other vertices. 
    
    We conclude that $G$ has twin-width more than $1$.
\end{proof}

\begin{theorem}
    We can decide in linear time if a given graph $G$ has twin-width at most $1$.
\end{theorem}

\begin{proof}
    We first compute a modular decomposition of $G$ and realisers for the quotient graphs of prime nodes \cite{McConnellS99}\footnote{In fact, it would be more reasonable to compute a permutation diagram and deduce a modular decomposition via common intervals for a practical implementation.}, or conclude that $G$ is not a permutation graph and therefore not a graph of twin-width $1$. It is well-known that the total size of the quotient graphs of prime nodes is linear in the size of the original graph.
    
    We then check all prime nodes with corresponding quotient graph $H$ as follows:
    
    Guess an extremal vertex $s$ (out of at most $4$) that can be the last vertex incident to the red edge in a $1$-contraction sequence of $H$, and then apply the recursive algorithm of \cref{lem:dec-algo}.

    If all prime nodes have twin-width $1$, we obtained a contraction sequence for all of them and they can be combined to form a contraction sequence of $G$ using \cref{coro:seq-modular}.

    All of the above procedures can be implemented to run in linear time.
\end{proof}

\section{Twin-width of distance-hereditary graphs}\label{sec:DH}

Distance-hereditary graphs are the graphs whose connected induced subgraphs preserve the distances between vertices. They are also graphs that can be fully decomposed by the split decomposition (see Theorem~\ref{thm:split-DH}).

The split decomposition is a graph decomposition introduced by Cunningham \cite{SplitDecIntro} based on recursively finding \emph{splits} of the graph, which are bipartitions $(A,B)$ of $V(G)$ with $|A|,|B| \geq 2$ such that the edges crossing $(A,B)$ form a biclique. The split decomposition of a graph can be computed in linear time \cite{SplitAlgo}. 

To describe split decompositions, we use the notion of graph-labelled trees introduced in \cite{GioanP12}, we begin by introducing their formalism and some of their results.

\begin{definition}[\hspace{-3pt} \cite{GioanP12}]
A \emph{graph-labelled tree} $(T,\mathcal{F})$ is a tree $T$ in which every node $v$ of degree $k$ is labelled by a graph $G_v \in \mathcal{F}$ on $k$ vertices, called \emph{marker vertices}, such that there is a bijection $\rho_v$ from the tree-edges of $T$ incident to $v$ to the marker vertices of $G_v$. If $\rho_v(e) = q$, then $q$ is called an \emph{extremity} of $e$.  
\end{definition}

Let $(T,\mathcal{F})$ be a graph-labelled tree and $l$ be a leaf of $T$. A node or a leaf $u$ different from $l$ is $l$-accessible if for every tree-edges $e = wv$ and $e'=vw'$ on the $l,u$-path in $T$, we have $\rho_v(e)\rho_v(e') \in E(G_v)$. By convention, the unique neighbour of the leaf $l$ in $T$ is also $l$-accessible.

\begin{definition}[\hspace{-3pt} \cite{GioanP12}]
The \emph{accessibility graph} of a \emph{graph-labelled tree} $(T,\mathcal{F})$ is the graph $Gr(T,\mathcal{F})$ whose vertex set is the leaf set of $T$, and in which there is an edge between x and y if and only if y is $x$-accessible. In this setting, we say that $(T,\mathcal{F})$ is a graph-labelled tree of $Gr(T,\mathcal{F})$.
\end{definition}

\begin{lemma}[\hspace{-3pt} \cite{GioanP12}]
Let $(T,\mathcal{F})$ be a graph-labelled tree. The accessibility graph $Gr(T,\mathcal{F})$ is connected if and only if for every node $v$ of $T$ the graph $G_v \in \mathcal{F}$ is connected.
\end{lemma}

\begin{lemma}[\hspace{-3pt} \cite{GioanP12}]
Let $(T,\mathcal{F})$ be a graph-labelled tree of a connected graph $G$ and let $v$ be a node of $T$. Then every maximal tree of $T-v$ contains a leaf $l$ such that $v$ is $l$-accessible.
\end{lemma}

The split decomposition of a graph $G$ gives a graph-labelled tree whose accessibility graph is $G$, where nodes are labelled by graphs that do not admit any split, and degenerate graphs which are cliques and stars.

\begin{theorem}[\hspace{-3pt} \cite{GioanP12}]\label{thm:split-DH}
A graph is distance-hereditary if and only if it is the accessibility graph of a graph-labelled tree labelled only by cliques and stars.
\end{theorem}

\begin{theorem}[\hspace{-3pt} \cite{GioanP12}]\label{thm:split-cograph}
A graph is a cograph if and only if it is the accessibility graph of a graph-labelled tree labelled only by cliques and stars, with all stars pointing to a fixed edge of the tree (such a tree corresponds exactly to the modular decomposition tree of the cograph).
\end{theorem}

We now give some simple results on the relation between twin-width and split decompositions.

\begin{lemma}\label{lem:ATfree-HD-dec}
The canonical split decomposition tree of a prime\footnote{for the modular decomposition} AT-free graph is a caterpillar.
\end{lemma}

\begin{proof}

Assume one of the internal nodes of the distance-hereditary decomposition is adjacent to at least $3$ other internal nodes of the decomposition.

We root our decomposition at this node. 
In each subtree, there must be an internal node labelled by a star whose center is not on the root's side or a graph that does not admit a split.
Indeed, if no such node is found in some subtree, then this subtree is a cograph, which contradicts primality.
In particular, we can safely replace a subtree by a node labelled by such a star because our graph must contain an induced subgraph described by such a decomposition tree. This is because the set of vertices that have the adjacency of the subtree's interface vertex in the root's label must have a splitter (otherwise it would be a nontrivial module).
However, we might need to extract more structure from a subtree in some cases.

We first consider the case when our root is labelled by a triangle. In this case, we find $A_{2,2}$ as an induced subgraph.

Then we consider the case when our root is labelled by a star and three of its branches lead to interface vertices. In this case, we find $A_1$ as an induced subgraph.

When the root is labelled by a star whose center is an interface vertex, the subtree of the decomposition corresponding to this interface vertex must contain more structure. Indeed, if it contains only a star the canonicity of the decomposition is contradicted.
In all cases, in the graph induced by this subtree of the decomposition, the interface vertex is adjacent to at least two vertices (otherwise the only vertex would be a cutvertex and we can collapse the node where it is introduced to the root) and these two vertices have a splitter (otherwise the graph is not prime). This means that we can extract an induced subgraph that is represented by two nodes: either a star with its center to the root and an attached star with its center away from the root, or a triangle and an attached star with its center away from the root. In these cases, we find $A_{2,3}$ or $A_3$ as an induced subgraph.

Finally, if the root is labelled by a graph that does not admit a split, in particular this graph is $2$-connected. This is sufficient to obtain an asteroidal triple: from each subtree of the decomposition, we obtain a pendant edge incident to this graph, any triple of the pendant vertices forms an asteroidal triple.

In all cases, we found an asteroidal triple.
We conclude that if the decomposition tree is not a caterpillar, the represented graph is not AT-free.
\end{proof}

\begin{corollary}
    We deduce the forbidden induced subgraphs characterising DH $\cap$ AT-free: Holes, Domino, House, Gem, $A_1$, $A_{2,2}$, $A_3$, $A_{2,3}$ (see Figure~\ref{fig:3asteroid}).
\end{corollary}

Since the classes of DH and AT-free graphs both already have a complete characterisation by forbidden induced subgraph, we could also compute the obstructions from these lists.

\begin{lemma}
The graphs that are AT-free and distance-hereditary are exactly the graphs that are permutation graphs and distance-hereditary.
\end{lemma}

\begin{proof}
Permutation graphs are AT-free so there is only one inclusion to show.

We first use the fact that both AT-free graphs and permutation graphs are stable by replacing a vertex with a graph of the class.

Consequently, it suffices to show the equivalence for prime graphs.
To do that, we proceed by induction on their decomposition tree whose structure is given in \cref{lem:ATfree-HD-dec}.

We prove by induction that a prime AT-free distance-hereditary graph has a realiser as a permutation graph where the leaves of the last node of the decomposition tree are extremal vertices.

As a base case, we have the graph with a single vertex, which has an obvious realiser as a permutation graph where it is an extremal vertex.

For the inductive case, we consider the three possible labels of the last node of the decomposition tree of our graph $G$ (the node cannot have degree more than 3 for a prime graph). Let $G'$ denote the graph represented by the decomposition minus this last node.

If the label is a star with a leaf as its center, we are exactly adding a pendant vertex to an extremal vertex of $G'$. We can take the realiser of $G'$ as a permutation graph and add a vertex that intersects exactly the corresponding extremal vertex to produce of a realiser of $G$ as a permutation graph.

If the label is a star with its center attached to the prefix of the decomposition tree, we are adding a false twin to an extremal vertex of $G$. If the label is a triangle, we are adding a true twin to an extremal vertex of $G$. In these two cases, we cannot guarantee that the two leaves are extremal vertices. However, since the vertices are twins, we can always assume that the leaf to which the rest of the decomposition tree will be attached is the one which is extremal in the realiser of $G$ as a permutation graph.

We covered all possible cases of labels, hence, by induction, the property always holds.
\end{proof}

\begin{lemma}
Graphs that are AT-free and distance-hereditary have twin-width at most 1. Furthermore, if such graphs are prime, we can immediately deduce a $1$-contraction sequence from the split decomposition.
\end{lemma}

\begin{proof}

We proceed by induction on the split decomposition tree of prime AT-free distance-hereditary graphs.

The graphs consisting of at most 4 vertices have twin-width at most 1.

Consider now a nontrivial prime AT-free distance-hereditary graph $G$.
The leaf nodes of the decomposition tree must be `pendant vertex' nodes because the graph is prime.

We assume that one of the pendant edges at one end of the decomposition tree is a red edge. Observe that making an existing edge red can only increase the twin-width of the trigraph. Hence, proving the result for such trigraphs will imply the result for graphs without a red edge.

We now consider the different cases of the decomposition node $u$ adjacent to the node corresponding to the pendant red edge. We denote by $a$ the pendant vertex, by $b$ its neighbour and by $c$ the leaf attached to $u$. $u$ has degree 3 since $G$ is prime and distance-hereditary.

\begin{itemize}
\item If $u$ is labelled by $C_3$, we identify $b$ and $c$. Indeed, the difference of their neighbourhoods is exactly $a$.
\item If $u$ is labelled by $P_3$ with $c$ as the central vertex, we identify $a$ and $b$. Indeed, $b$ has only one extra neighbour so this identification produces only one red edge and removes the previous pendant red edge.
\item If $u$ is labelled by $P_3$ with the rest of the decomposition tree attached to its central vertex, we identify $b$ and $c$. Indeed the difference of their neighbourhood is exactly $a$.
\item If $u$ is labelled by $P_3$ with the pendant red edge attached to its center, $a$ and $c$ are twins, this contradicts $G$ being prime.
\end{itemize}

In possible cases, we reduced the number of nodes of the decomposition due to the identification. After possibly identifying twins, we obtain a graph that is prime, distance-hereditary, with a pendant red edge, and with strictly less nodes in the decomposition. We apply the induction hypothesis to this graph to conclude that there exists a $1$-contraction sequence.
\end{proof}

The following result is a simple extension of the construction of $2$-contraction sequences for trees.

\begin{lemma}
    DH graphs have twin-width at most $2$. We can compute a $2$-contraction sequence for them in linear time.
\end{lemma}

\begin{proof}
    We follow an elimination sequence of twins and pendant vertices of the DH graph to produce a contraction sequence such that, for each trigraph, its underlying graph is still DH, at most one vertex may be incident to two red edges and all red edges are pendant.

    If a vertex is incident to two pendant red edges, we may simply contract the corresponding two pendant vertices. We may now assume that at most one red edge is incident to each vertex of the trigraph and that all red edges are pendant.
    
    We now consider the DH graph $H$ induced by the black edges of the trigraph, $H$ has a pair of twins or a pendant vertex. We may contract the twins which produces a vertex of red degree at most two, and the red edges remain pendant. In case of a pendant vertex $v$ in $H$, we contract its incident red edge if it exists, and either way, we now consider the black edge incident to $v$ to become red.
    
    We will fully eliminate the black part of the trigraph, leaving a single red edge which we can then contract. Indeed, the black part of the trigraph sequence corresponds exactly to the sequence of DH graphs corresponding to the elimination of twins and pendant vertices.

    Regarding the implementation of this contraction sequence, we can compute the split decomposition in linear time. It is well-known that we can then obtain an elimination sequence in linear time.
\end{proof}

\begin{theorem}
    If $G$ is a distance-hereditary graph, the value of its twin-width is:
    \begin{itemize}
        \item 0 if it is a cograph
        \item 1 if it is AT-free, or, equivalently, if it is a permutation graph
        \item 2 otherwise
    \end{itemize}
\end{theorem}

We deduce the following from known algorithms.

\begin{corollary}
    If $G$ is a distance-hereditary graph, we can decide the twin-width of $G$ in linear time.
\end{corollary}

Note that we do not strictly require $G$ to be a DH graph and could simply answer either that $G$ is not DH or give the value of the twin-width.

We may also deduce the value of the twin-width directly from the structure of the split decomposition by combining Theorem~\ref{thm:split-cograph}, and Lemma~\ref{lem:ATfree-HD-dec}. Recall that only stars and cliques appear in the graph-labelled tree of a distance-hereditary graph.

\begin{theorem}
    If $G$ is a distance-hereditary graph, the value of its twin-width is:
    \begin{itemize}
        \item 0 if, for each of its connected components, there is an edge of the corresponding split decomposition towards which all stars labelling nodes of the decomposition point.
        \item 1 if, for each of its connected components, there is a path of the split decomposition such that all stars labelling nodes of the decomposition point to the path\footnote{There is no constraint for stars on the path}.
        \item 2 otherwise
    \end{itemize}
\end{theorem}

\bibliography{references}
\appendix

\begin{figure}[t]
    \centering
    \includegraphics[scale=0.75]{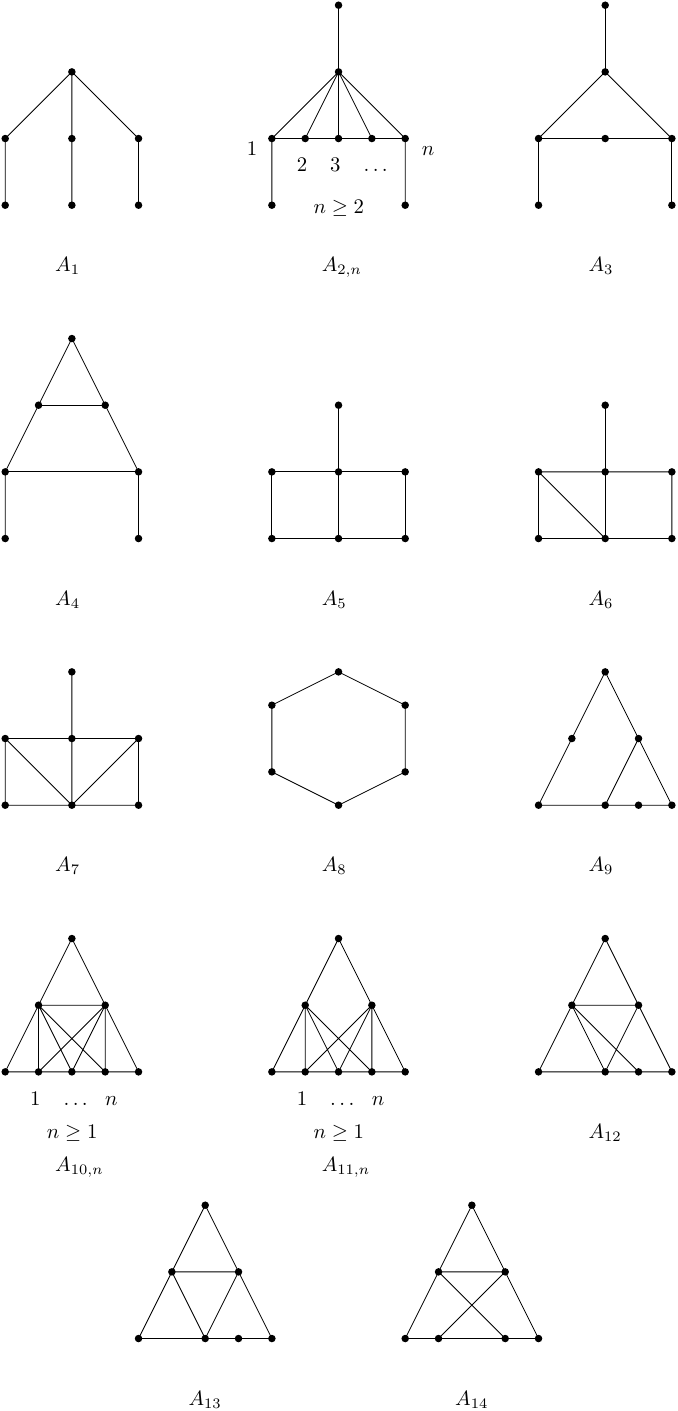}
    \caption{Minimal graphs containing an asteroid triple.}
    \label{fig:3asteroid}
\end{figure}

\begin{figure}[t]
    \centering
    \includegraphics[scale=0.7]{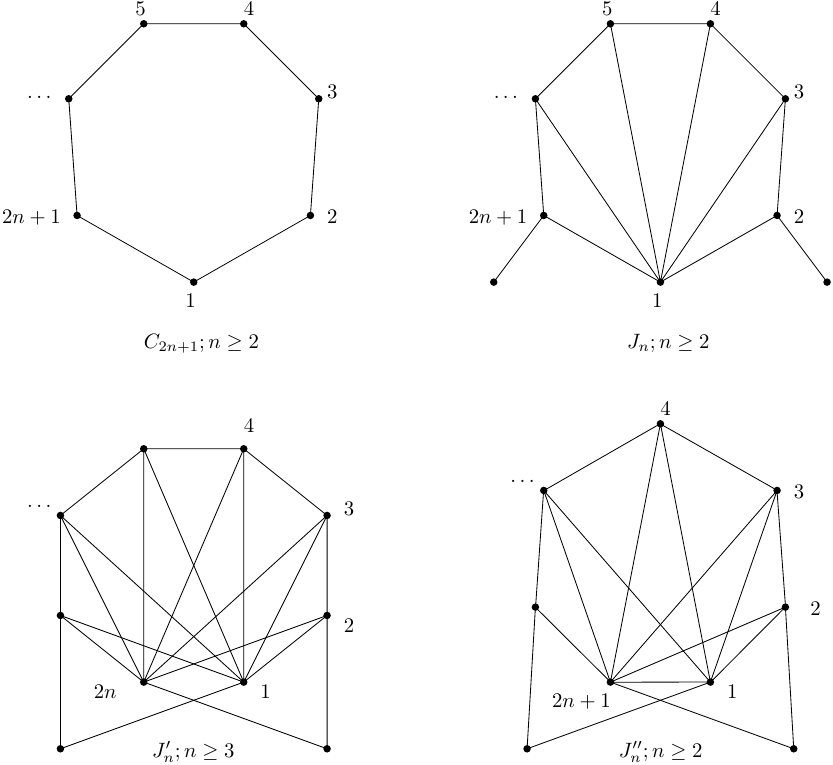}
    \caption{Complements of minimal graphs containing a $(2k+1)$-asteroid for $k>1$.}
    \label{fig:largeasteroid}
\end{figure}

\end{document}